\documentclass[conference,letterpaper]{IEEEtran}
\usepackage[top=0.73in,bottom=1.1in,left=0.637in,right=0.64in]{geometry}
\IEEEoverridecommandlockouts
\usepackage[utf8]{inputenc}
\usepackage[T1]{fontenc}
\usepackage{url}
\usepackage{ifthen}
\usepackage{cite}
\usepackage[cmex10]{amsmath}
\usepackage{mathtools}
\usepackage{bbm}
\usepackage{graphicx}
\usepackage{booktabs}
\usepackage{dsfont}
\usepackage{amssymb}
\usepackage{amsthm}
\usepackage{amsfonts}
\usepackage{hyperref}
\usepackage[hyphenbreaks]{breakurl}
\usepackage{cleveref}
\usepackage{xcolor}
\usepackage{enumitem}

\definecolor{rc}{RGB}{139,0,0}   

\newcommand{\AUT}{\operatorname{AUT}}
\newcommand{\Ham}{\operatorname{w_H}}
\newcommand{\dH}{\operatorname{d_H}}
\newcommand{\DB}{\operatorname{DB}}
\newcommand{\B}{\operatorname{B}}
\newcommand{\RM}{\operatorname{RM}}
\newcommand{\supp}{\operatorname{supp}}
\newcommand{\vect}[1]{\underline{#1}}
\newcommand{\AGL}{\operatorname{AGL}}

\newenvironment{talign*}
 {\csname align*\endcsname}
 {\endalign}

\date{}

\newtheorem{lemma}{Lemma}
\newtheorem{observation}{Observation}
\newtheorem{corollary}{Corollary}

\newtheorem{definition}{Definition}
\newtheorem{theorem}{Theorem}

\newtheorem{note}{Note}

\title{On the Automorphism Groups of Berman Codes and associated Abelian Codes}
\author{Harshvardhan Pandey \and Prasad Krishnan
\thanks{Pandey and Krishnan are with the Signal Processing and Communications Research
Center, International Institute of Information Technology, Hyderabad 500032, India
(emails: \{harshvardhan.pandey@research., prasad.krishnan@\}iiit.ac.in).
\textit{Acknowledgement:} The authors would like to thank Lakshmi Prasad Natarajan
for fruitful discussions on this work.}}

\begin{document}
\maketitle
\vspace{-1.5cm}

\begin{abstract}
The automorphism group of a code is the group of permutations that map a code to
itself. Berman codes are a class of binary linear codes characterized by two integer
parameters $n\geq 2$ and $m\geq 1$, and this class includes the Reed-Muller codes
as well. The class of Berman codes and their duals were recently shown to achieve
the capacity of the binary erasure channel.
A number of abelian codes that arise from the intersection and subspace sums of
Berman and Dual Berman codes were also identified recently, for odd $n\geq 3$. A
subclass of these abelian codes was shown to have good short block-length performance
for AWGN channels, with efficient decoding algorithms. In this work, we identify the
exact automorphism group for Berman codes and their duals. Further, we find the
exact automorphism group for the abovementioned abelian codes, when $n\geq 5$. In
the case of such abelian codes with $n=3$, we present partial characterizations of
the automorphism groups for a large collection of parameter choices, and complete
characterizations for a few.
\end{abstract}

\section{Introduction}

Reed-Muller (RM) codes \cite{Muller1954RMCodes, Reed1954RMCodes} are an important
class of algebraic codes which have been extensively studied in coding theory.
The automorphism group of the code ${\cal C}$ of block length $N$ is the subgroup of $S_N$ (the group
of all permutations of $N$ items) which preserves the code,
i.e., $\pi \in \AUT(\cal C)$ if and only if $\left \{\pi(\vect{c}):
~\vect{c} \in \cal C \right\} = \cal C$. It is well known that any linear code
${\cal C}$ and its dual ${\cal C}^\perp$ share the same automorphism group. For
$1 \le r \le m-2$, the automorphism group of RM codes is isomorphic to the affine
general linear group $\AGL(m, 2)$, and for $r \in \{0, m-1, m\}$, the automorphism
group is $S_{2^m}$. Using the fact that the automorphism group is doubly transitive,
RM codes were shown to achieve capacity on the binary erasure channel
\cite{RM_BEC_Capacity}. It was later shown \cite{RM_BMS_Capacity} that RM codes
achieve capacity on any binary memoryless symmetric channel using the automorphism
group and some additional properties. Decoding algorithms that take advantage of
the large automorphism group of RM codes have also been proposed
\cite{automorphism_ensemble_decoding, constituent_automorphism_decoding} which
perform close to the maximum likelihood decoder, while also being computationally
tractable.

Berman introduced a class of codes of block length $p^m$, for any odd prime $p$,  which can be studied as ideals of a group algebra \cite{berman1967semisimple}. Blackmore and Norton studied
a class of codes of block length $n^m$, for all $n \ge 2$
\cite{blackmoreDualBerman}. These codes strictly contain RM codes (when $n=2$), as
well as the duals of Berman codes. Later, Natarajan and Krishnan showed that Berman
codes and their duals also achieve capacity on the binary erasure channel
\cite{natarajan2023bermancodesgeneralizationreedmuller}. Further, a large class of
abelian codes of length $n^m$ (for odd $n$) was also presented in this work. An
abelian code is an ideal in a group algebra. In the case of
\cite{natarajan2023bermancodesgeneralizationreedmuller}, the group algebra may be
taken as ${\mathbb Z}_n^m$. Using the finite-field discrete fourier transform (DFT)
characterization \cite{10.1109/18.165458} of abelian codes, various classes of
binary abelian codes, defined by the integers $n$ (odd), $m$, and a subset
${\cal W}\subseteq \{0,\hdots,m\}$ denoting the collection of non-zero ``frequencies'' of the DFT of the code, were presented. These abelian codes are closed under
duality. Berman codes and their duals can also be viewed from this standpoint when
$n$ is odd. Some of these classes were also shown to achieve the binary erasure
channel capacity. This was done partly by identifying a subgroup of the automorphism
group of those codes.

The automorphism group of only a few other classes of codes are completely known.
For example, the automorphism group of Hamming codes is known to be the general
linear group with its natural action on non-zero vectors. This follows from the fact
that we can view Hamming codes as punctured RM codes. Another important example is
that of Golay codes. The automorphism group of the perfect binary Golay code
$\mathcal{G}_{23}$ is the Mathieu group $M_{23}$, while that of the extended binary
Golay code $\mathcal{G}_{24}$ is $M_{24}$
\cite[Chapter \,20, Corollary \,6\,\&8]{MacWilliams1977}. Among non-binary codes, automorphisms have been studied for Cauchy Reed Solomon codes
\cite{ReedSolomonAutomorphism}, generalized RM codes \cite{GRM_Automorphism} and
the ternary Golay code \cite[Chapter ~20, Theorem 18]{MacWilliams1977}. In some
cases only subgroups of the automorphism group are known. For example, for polar
codes Geiselhart et al.~\cite{Gei21} showed that the automorphism group contains
the block lower-triangular affine group (BLTA). To the best of our knowledge, the
full automorphism group of polar codes is not known in general.

In this work we study the automorphism groups of dual Berman codes (and by duality,
Berman codes) by reduction to graph automorphism. Furthermore using results in
maximal subgroups of the symmetric group \cite{LIEBECK1987365,
10.1093/qmath/37.4.419} we present various results regarding the automorphism of
large classes of related abelian codes. We now summarize our contributions and
organization of the paper.
\begin{itemize}[leftmargin=*]
    \item \textit{\Cref{sec:group_theory}:} We define and discuss some of the group
    theoretic preliminaries required for the paper.
    \item \textit{\Cref{sec:mainresults}:} We give a brief overview about Berman
    codes and their duals, along with some properties. We further define the abelian
    codes studied in this work other than Berman/Dual Berman codes. This is followed
    by the main results of the paper, given by Theorems
    \ref{thm:maintheoremdualberman}~--~\ref{thm:n=3,singleton-weight-set}.
    \item \textit{\Cref{sec:proofofAutofBerman}:} We prove
    \Cref{thm:maintheoremdualberman} in this section. The theorem precisely
    characterizes the complete automorphism group of all Berman codes with $n > 2$,
    i.e., the non RM cases.
    \item \textit{\Cref{sec:proofofTheoremAUTn>=5}:} \Cref{thm:n>=5AUT} is proved
    in this section. The associated abelian codes studied in the paper have a
    parameter $n \ge 3$ which is an odd integer. Using the previous result of Berman
    codes and existing literature about maximal subgroups of the symmetric group,
    \Cref{thm:n>=5AUT} characterizes the complete automorphism group of all abelian
    codes in the class we consider in this work, having $n \ge 5$.
    \item \textit{\Cref{sec:n=3}:} This section looks at our considered abelian
    codes with $n=3$. The $n=3$ case turns out to be far more involved than the
    prior two scenarios. We present partial characterizations of the automorphism
    group for these abelian codes for a large number of parameter choices, and
    complete characterization in a few select settings. Specifically, we prove
    \Cref{theorem:affine_type_permutation}, which states that the automorphism group
    is a subgroup of the affine general linear group. We also prove Theorem
    \ref{thm:n=3,singleton-weight-set}, which looks at the special case when the
    allowed weight set corresponding with the abelian code is a singleton set, i.e.,
    permitting only a single weight for the finite-field DFT of the code.
\end{itemize}
Our techniques in this work exploit the rich collection of results involving the
maximal subgroups of the permutation group. We present various group-theoretic
properties of such subgroups along the way to characterize the automorphisms of our
codes. We believe that this approach presents a potential foundation for analyzing
the automorphism group of many other abelian codes as well.

\emph{Notation and basic terminology:} For any positive integer $n$, let $[n]$
denote the set $\{0,1,\dots,n-1\}$. For a subset $A\subseteq B$, its complement
$B\setminus A$ is denoted by $\overline{A}$. Let $\mathbb{Z}_{n}$ be the ring where
addition and multiplication are performed modulo $n$.
$\mathbb{F}_q$ denotes the field with $q$ elements. The identity matrix of order
$n$ is denoted by $I_n$. The empty set is $\emptyset$. Vectors are denoted with an
underline, and considered as column vectors unless otherwise stated. For any vector
$\vect{u}$, $\supp(\vect{u})$ is the support set of $\vect{u}$. The Hamming weight
of a vector ${\vect{u}}$ is denoted by $\Ham({\vect{u}})$. The all ones vector of length $n$ is denoted by $\mathbbm{1}_n$. The Hamming distance
between vectors $\vect{u}$, and $\vect{v}$ is denoted by $\dH(\vect{u}, \vect{v})$.
$S_n$ denotes the group of all permutations of a collection of $n$ items, under the
operation of composition. A linear code $\mathcal{C}$ of length $n$ over
$\mathbb{F}_q$ is a subspace of $\mathbb{F}_q^n$. The automorphism group of
$\mathcal{C}$, denoted by $\AUT(\mathcal{C})$, is the subgroup of $S_n$ which
preserves the code. The dual code of ${\cal C}$ is given by ${\cal
C}^\perp=\{\underline{c_1}\in{\mathbb F}_q^n:\underline{c_1}^T\underline{c}=0,
\forall \underline{c}\in{\cal C}\}$.

\section{Group Theory Preliminaries}
\label{sec:group_theory}

This section provides a brief overview of some of the group theoretic ideas we will
be using. We refer the reader to \cite{dixon1996permutation} for a comprehensive
treatment. Let $G$ be a group. The notation $H \le G$ denotes that $H$ is a
subgroup of $G$, and $H < G$ indicates that $H$ is a proper subgroup. If $H$ is a
normal subgroup of $G$, it is denoted by $H \unlhd G$. We will call a group trivial
if it contains only one element (the identity).

\begin{definition}[Index of a Subgroup]
Let $H \le G$. The left cosets of $H$ in $G$ are sets of the form $gH = \{gh : h
\in H\}$ for $g \in G$. The right cosets of $H$ are similarly defined, and denoted
as $Hg, g\in G$. The index of $H$ in $G$, denoted by $[G : H]$, is the number of
distinct left (or right) cosets of $H$ in $G$. By Lagrange's Theorem, for finite
groups, $|G| = [G : H] |H|$.
\end{definition}

\begin{definition}[Group Action]
Let $G$ be a group with identity element $e$, and let $\Omega$ be a set. A (left)
group action of $G$ on $\Omega$ is a function from $G \times \Omega$ to $\Omega$,
typically denoted by $(g, \omega) \mapsto g ( \omega)$, that satisfies two
fundamental axioms. First, the identity axiom states that $e ( \omega) = \omega$
for all $\omega \in \Omega$. Second, the compatibility axiom requires that $g \left(
h (\omega) \right) = (gh)( \omega)$ for all $g, h \in G$ and $\omega \in \Omega$.
For $\Delta\subseteq \Omega$, the collection of its images under the action by
$g\in G$ is denoted as $g(\Delta)$.
\end{definition}

\begin{definition}[Faithful Group Action]
An action of a group $G$ on a set $\Omega$ is called faithful, if the only element
of $G$ that fixes every element in $\Omega$ is the identity element $e$.
\end{definition}

In this work, we study the permutation automorphisms of binary linear codes. Since
the natural action of the symmetric group $S_N$ on a set $\Omega$ of size $N$ is
faithful, any subgroup of $S_N$ would also act faithfully on $\Omega$. Thus, we
will only look at faithful group actions.

\begin{definition}[Orbit]
Let a group $G$ act on a set $\Omega$. The orbit of an element $\omega \in \Omega$
under this action, denoted by $G ( \omega)$, is the set of all possible destinations
of $\omega$ under the action of $G$. Formally, $G ( \omega) = \{g ( \omega) : g \in
G\}$. The orbits of elements of $\Omega$ form a partition of the set $\Omega$.
\end{definition}

\begin{definition}[Symmetric and Alternating Groups]
The symmetric group on a set of $n$ elements, denoted by $S_n$, is the group of
all permutations of that set under function composition. The sign of a permutation
$\sigma \in S_n$, denoted $\operatorname{sgn}(\sigma)$, is $+1$ if $\sigma$ can be
expressed as a product of an even number of transpositions (swap operations), and
$-1$ if it requires an odd number. The set of all permutations with a sign of $+1$
forms a normal subgroup of $S_n$ called the alternating group, denoted by $A_n$,
which has order $|S_n|/2$, for $n \ge 2$.
\end{definition}

\begin{definition}[Transitivity]
A group action of $G$ on $\Omega$ is transitive if, for any pair of elements
$\alpha, \beta \in \Omega$, there exists some $g \in G$ such that $g(\alpha) =
\beta$. It is doubly transitive if any ordered pair of distinct elements can be
mapped to any other ordered pair of distinct elements. The stabilizer of an element
$\omega \in \Omega$ is the subgroup $G_\omega = \{g \in G : g(\omega) = \omega\}$.
\end{definition}

\begin{definition}[Primitivity]
For a group $G$ acting on a set $\Omega$, a subset $\Delta \subseteq \Omega$ is
called a block under $G$ if for every $g \in G$, either $g(\Delta) = \Delta$ or
$g(\Delta) \cap \Delta = \emptyset$. The sets $\Omega$, the empty set, and singleton
sets are trivial blocks. A transitive permutation group is said to be primitive if
it preserves no non-trivial blocks.
\end{definition}

\begin{definition}[Automorphisms of a Group]
An automorphism of $G$ is a bijective function $\alpha : G \to G$ such that for all
$a, b \in G$, $\alpha(ab) = \alpha(a)\alpha(b)$. The set of all automorphisms of
$G$ is denoted by $\AUT(G)$ which is itself a group under  the composition operation.
\end{definition}

\begin{definition}[Semidirect Product]
Let $N$ and $H$ be groups, and let $\varphi : H \to \AUT(N)$ be a group
homomorphism. The semidirect product $N \rtimes_{\varphi} H$ is a group constructed
on the Cartesian product $N \times H$ with the group operation defined by $(n_1,
h_1)(n_2, h_2) = (n_1 \varphi(h_1)(n_2), h_1 h_2)$.
\end{definition}

\begin{definition}[Wreath Product]
Let $G$ be a group and $H \le S_m$ be a permutation group. The wreath product of
$G$ by $H$, denoted $G \wr H$, is the semidirect product $G^m \rtimes_{\varphi} H$,
where the base group is the $m$-fold direct product $G^m$, and the homomorphism
$\varphi : H \to \operatorname{Aut}(G^m)$ is the coordinate permutation action given
by $\varphi(h)(g_1, \dots, g_m) = (g_{h^{-1}(1)}, \dots, g_{h^{-1}(m)})$.
\end{definition}

\begin{definition}[Product Action of the Wreath Product]
The product action of the wreath product $S_n \wr S_m$ on the set $[n]^m$ is
defined as follows. An element of $S_n \wr S_m$ is given by a tuple $((\sigma_1,
\dots, \sigma_m), \pi)$, where $\sigma_i \in S_n$ and $\pi \in S_m$. Its action on
a vector $\vect{x} = (x_1, \dots, x_m) \in [n]^m$ is given by $((\sigma_1, \dots, \sigma_m), \pi) \cdot (x_1, \dots, x_m) = (\sigma_1(x_{\pi^{-1}(1)}), \dots, \sigma_m(x_{\pi^{-1}(m)}))$.
\end{definition}

\begin{definition}[Affine General Linear Group]
Let $\mathbb{F}_q$ denote the finite field of $q$ elements. The general linear
group, $\operatorname{GL}(m, q)$, is the group of all invertible $m \times m$
matrices over $\mathbb{F}_q$. The affine general linear group,
$\operatorname{AGL}(m, q)$, consists of all affine transformations of the vector
space $\mathbb{F}_q^m$. The elements of $\operatorname{AGL}(m, q)$ are mappings of
the form $\vect{x} \mapsto A\vect{x} + \vect{b}$, where $A \in \operatorname{GL}(m,
q)$ and $\vect{b} \in \mathbb{F}_q^m$.
\end{definition}

\begin{definition}[Simple Group]
A group $G$ is called simple if it is non-trivial and its only normal subgroups are
the trivial group and the group $G$ itself.
\end{definition}

\begin{definition}[Almost Simple Group]
\label{defn:almostsimple}
A group $G$ is said to be almost simple if there exists a non-abelian simple group
$S$ such that $S \triangleleft G \leq \AUT(S)$, where $\AUT(S)$ is the automorphism
group of $S$.
\end{definition}

\begin{definition}[Solvable Group]
A group $G$ is solvable if there exists a finite sequence of subgroups $\{e\} = G_0
\triangleleft G_1 \triangleleft \dots \triangleleft G_k = G$ such that for each $i
\in [k]$ the quotient group $G_{i+1}/G_i$ is an abelian group.
\end{definition}

Another important class of permutation groups are \textit{groups of diagonal type}.
The construction of these groups require developing some framework which is outside
the scope of this work. We refer the reader to \cite[Section ~4.5]{dixon1996permutation} for a more in depth review. For the
purpose of this work, we will only need one property of groups of diagonal type. If
such a group acts primitively on a set $\Omega$, then $|\Omega| = |T|^{l-1}$ for
some non-abelian simple group $T$ \cite[Theorem ~4.1A(b)(ii)]{dixon1996permutation}.
We will now present a useful lemma. We include a proof for the sake of completeness.

\begin{lemma}
    \label{lemma:normal-subgroup-of-primitive-group-is-transitive}
    If $G$ acts primitively and faithfully on $\Omega$, and $H \unlhd G$ such that
    $|H| > 1$, then $H$ acts transitively on $\Omega$.
\end{lemma}

\begin{proof}
    First, we note that $|\Omega|\geq 2$ since $G$ acts faithfully. Let
    $\Delta_\alpha=\{h\alpha \mid h \in H\}$ for some $\alpha \in \Omega$. We will
    show that $\Delta_\alpha$ is a block for the action of $G$ on $\Omega$. For any
    $g \in G$, we have $g(\Delta_\alpha)=\{(gh)(\alpha) \mid h \in H\}$. Since $H
    \unlhd G$, for every $h \in H$, there exists some $h' \in H$ such that $gh=h'g$.
    Thus, we have $g(\Delta_\alpha)=\{(h'g)\alpha \mid h' \in H\}$. Thus,
    $g(\Delta_\alpha)$ is exactly the orbit of $g(\alpha)$ under $H$. Because the
    orbits of $H$ form a partition of $\Omega$, the orbit $g(\Delta_\alpha)$ must
    either be identical to $\Delta_\alpha$ or completely disjoint from it. Therefore,
    $\Delta_\alpha$ is a block under $G$. Since $G$ acts primitively on $\Omega$,
    the only blocks are trivial. Either $|\Delta_\alpha|=1$ or $\Delta_\alpha=\Omega$.
    If $|\Delta_\alpha|=1$ for any $\alpha \in \Omega$, then $|\Delta_\beta|=1 \,,
    \forall \beta \in \Omega$. This means every element of $H$ fixes every element
    of $\Omega$. This is not possible since $|H| > 1$, and action of $G$ is faithful.
    It follows that $\Omega$ has exactly one orbit under $H$, i.e., we have
    $\Delta_\alpha=\Omega, \forall \alpha$. Therefore, $H$ acts transitively on
    $\Omega$.
\end{proof}

\section{Relevant Literature Review and Main Results}
\label{sec:mainresults}
In this section, we recollect the formal definitions for Berman codes and their duals. Further, we also present the definitions for the related Abelian codes studied in this work. This is followed by the main theorems and their proof overviews. 
\subsection{Review of Berman Codes, their duals, and related abelian codes}
\label{subsec:reviewofBermanANDBiDcodes}

\subsubsection{Berman codes and their duals \cite{berman1967semisimple,natarajan2023bermancodesgeneralizationreedmuller}}
\label{subsubsec:bermancodesandduals}

For an integer $n \ge 2$ and $m \ge 1$, Berman codes of order $m$ and $0$ are
given by $\B_n(m,m) \triangleq \{{\vect{0}}\in \mathbb{F}_2^{n^m}\}$ and
$\B_n(0,m) \triangleq \{{\vect{c}}\in \mathbb{F}_2^{n^m}\colon \sum_i c_i=0\}$
respectively. For $1 \le r \le m-1$,
\begin{align*}
    \B_n(r,m) &\triangleq \{ ({\vect{v}}_0|{\vect{v}}_1|\hdots|{\vect{v}}_{n-1})
    \colon {\vect{v}}_l\in \B_n(r-1,m-1),\\
    &\qquad \sum_{l\in [n]}{{\vect{v}}_l}\in \B_n(r,m-1) \}.
\end{align*}
For an integer $n \ge 2$ and $m \ge 1$, Dual Berman codes of order $m$ and $0$ are
given by $\DB_n(m,m) \triangleq \mathbb{F}_2^{n^m}$ and $\DB_n(0,m) =
\{(c,\hdots,c)\in \mathbb{F}_2^{n^m}\colon c\in \mathbb{F}_2\}$ respectively. For
$1 \le r \le m-1$,
\begin{align*}
    \DB_n(r,m) &\triangleq \{
    ({\vect{u}}+{\vect{u}}_0|{\vect{u}}+{\vect{u}}_1|\hdots|{\vect{u}}+{\vect{u}}_{n-2}|{\vect{u}})
    \colon \\
    &\qquad {\vect{u}}_l\in \DB_n(r-1,m-1), \forall l \in [n-1],\\
    &\qquad {\vect{u}}\in \DB_n(r,m-1)\}.
\end{align*}
Note that these constructions can be seen as a generalization of the Plotkin
construction of RM codes. The codes $\DB_n(r,m)$ and $\operatorname{B}_n(r,m)$ are
known \cite{natarajan2023bermancodesgeneralizationreedmuller} to be duals of each
other.
$\DB_n(r,m)$ can also be constructed using recursive subproducts \cite{siddheshwar2024recursivesubproductcodesreedmullerlike}.

The automorphism group $\AUT({\cal C})$ of an $n^m$-length code $\mathcal{C}$ is
the subgroup of $S_{n^m}$ which preserves the code, i.e. if $\pi$ is an
automorphism of $\mathcal{C}$, then $\pi \left(\mathcal{C}\right) =
\left\{\pi(\vect{c}): \vect{c} \in \mathcal{C}\right\} = \mathcal{C}$. Note that
the automorphism group of a linear code ${\cal C}\subseteq {\mathbb F}_q^n$ is the
same as that of its dual code.
When $n = 2$, $\DB_2(r,m) = \RM(r,m)$. It is known that $\AUT(\RM(r,m)) =
\operatorname{AGL}(m,2)$ for $1 \le r \le m-2$, and $\AUT(\RM(r,m)) = S_{2^m}$ for
$r \in \{0, m-1, m\}$. It is also known
\cite[Theorem 14]{natarajan2023bermancodesgeneralizationreedmuller} that $S_n\wr
S_m\leq \AUT(\DB_n(r,m))= \AUT(\operatorname{B}_n(r,m))$.

For a linear code, a codeword $\underline{c}$ is called a minimum weight codeword
if $\Ham(\underline{c})$ is the minimum distance of the code. For $n > 2$, the
minimum weight codewords of $\DB_n(r, m)$ are given by the tensor product $\vect{u}
= \vect{u}_1 \otimes \vect{u}_2 \otimes \dots \otimes \vect{u}_m$, where $r$ of
the component vectors are standard basis vectors of $\mathbb{F}_2^n$, and the
remaining $m-r$ components are $\mathbbm{1}_n$
~\cite[Lemma~4]{siddheshwar2024recursivesubproductcodesreedmullerlike}. The
coordinates of the codewords are indexed by the set $[n]^m$.

\subsubsection{Berman-intersection-dual Berman (BiD) codes and related Abelian
Codes \cite{natarajan2023bermancodesgeneralizationreedmuller,
10.1109/18.165458,dash2025bidcodesalgebraiccodes}}
\label{subsubsec:abeliancodesreview}

For an odd integer $n \ge 3$ and $m \ge 1$, let $\mathcal{W} \subseteq \{0, 1,
\dots, m\}$ be a set of weights. We define the code $C(n, m, \mathcal{W})$ as an
ideal of the group algebra $\mathbb{F}_2[\mathbb{Z}_n^m]$ such that the DFT of a
codeword is supported on frequencies whose Hamming weight lies in $\mathcal{W}$.
Formally,
\[
\mathcal{C}(n, m, \mathcal{W}) = \left \{ \vect{a} \in \mathbb{F}_2[\mathbb{Z}_n^m] :
\widehat{a}_{\vect{j}} = 0 \text{ if } \operatorname{wt}_H(\vect{j}) \notin
\mathcal{W} \right \},
\]
where $\widehat{a}_{\vect{j}}$ denotes the $\vect{j}$-th component of
the DFT of $\vect{a}$. Such codes come under the class of abelian (group)
codes \cite{10.1109/18.165458}. It is easy to show that any $\mathcal{C}(n,m,
\mathcal{W})$ can be decomposed as
\[
\mathcal{C}(n,m,\mathcal{W}) = \bigoplus_{w \in \mathcal{W}} \mathcal{C}(n,m,\{w\}).
\]
The dual code of $\mathcal{C}(n,m,\mathcal{W})$ is
$\mathcal{C}(n,m,\overline{\mathcal{W}})$~\cite[Lemma~27]{natarajan2023bermancodesgeneralizationreedmuller}.

If ${\cal W}=\{0,1,\hdots,r\}$ for some $r\leq m$, then it can be seen
\cite{natarajan2023bermancodesgeneralizationreedmuller,
dash2025bidcodesalgebraiccodes} that $\mathcal{C}(n,m,{\cal W})=\DB_n(r,m)$, while
$\mathcal{C}(n,m,\overline{\cal W})=\operatorname{B}_n(r,m)$. Following this, if
${\cal W}=\{r_1,\hdots,r_2\}\subseteq[m+1]$ for some $r_1\leq r_2$, then we can
show \cite{natarajan2023bermancodesgeneralizationreedmuller,
dash2025bidcodesalgebraiccodes}, that $\mathcal{C}(n,m,{\cal
W})=\operatorname{B}_n(r_1-1,m)\cap \DB_n(r_2,m)$. Hence, when
${\cal W}=\{r_1,\hdots,r_2\}\subseteq[m+1]$, the code $\mathcal{C}(n,m,\mathcal{W})$
is termed as a
\textit{Berman-intersection-dual Berman (BiD) code}
\cite{dash2025bidcodesalgebraiccodes}. Properties of BiD codes with $n=3$ (and
thus of length $3^m$), their decoding algorithms, and comparisons with RM and Polar
codes were presented in \cite{dash2025bidcodesalgebraiccodes}. In general, for any
${\cal W}\subseteq [m+1]$, we have
\begin{align}
\label{eqn:decomposition}
    \mathcal{C}(n,m,\mathcal{W}) = \bigoplus_{w \in \mathcal{W}}
    \left(\DB_n(w, m) \cap \B_n(w-1,m)\right).
\end{align}

\begin{note}
    \label{note:trivial-codes-case}
    We shall refer to the following weight sets as \underline{boundary cases}.
    \begin{enumerate}
        \item $\mathcal{W} = \emptyset$, i.e., $\mathcal{C}(n, m, \mathcal{W}) =
        \{\vect{0}\}$.
        \item $\mathcal{W} = \{0\}$, i.e., $\mathcal{C}(n, m, \mathcal{W})$ is
        the repetition code of length $n^m$.
        \item $\mathcal{W} = \{1, 2, \dots, m\}$, i.e., $\mathcal{C}(n, m,
        \mathcal{W})$ is the single parity check code of length $n^m$.
        \item $\mathcal{W} = \{0, 1, \dots, m\}$ i.e., $\mathcal{C}(n, m,
        \mathcal{W}) = \mathbb{F}_2^{n^m}$.
    \end{enumerate}
    It is easy to verify that in these four cases $\AUT(\mathcal{C}(n, m, \mathcal{W})) = S_{n^m}$.
\end{note}

\subsection{Main Results}
\label{subsec:mainresults}

This work is concerned with the automorphism groups of the Berman codes and the abelian codes described in
Subsection \ref{subsec:reviewofBermanANDBiDcodes}. The main results of this work
are Theorems~\ref{thm:maintheoremdualberman}--\ref{thm:n=3,singleton-weight-set}, which we now present.

\begin{theorem}
    \label{thm:maintheoremdualberman}
    For $1 \le r \le m-1$ and $n\geq 3$, $\AUT \left(\DB_n(r, m)\right)
    =\AUT(\operatorname{B}_n(r,m))=S_n \wr S_m$. If $r \in \{0, m\}$, $\AUT
    \left(\DB_n(r, m)\right) =\AUT(\operatorname{B}_n(r,m))=S_{n^m}$.
\end{theorem}

We prove Theorem \ref{thm:maintheoremdualberman} in Section
\ref{sec:proofofAutofBerman}, by characterizing the necessary properties of
automorphism groups of Berman codes and their duals. Essentially, any
automorphism of the code $\DB_n(r,m)$ for $r\in[m-1]$ must be an automorphism of the Hamming
graph. Then, from a known result \cite{praeger2018permutation-cartesian-decomp}
regarding the automorphism group of the Hamming graph, we have Theorem
\ref{thm:maintheoremdualberman}.

\begin{theorem}
    \label{thm:n>=5AUT}
    For odd $n \ge 5$, $\AUT \left( \mathcal{C}(n, m, \mathcal{W}) \right) = S_n
    \wr S_m$, for any ${\cal W}\subseteq [m+1]$ such that ${\cal W}$ is not a
    boundary case (as in \Cref{note:trivial-codes-case}). For ${\cal W}$ being any
    boundary case, the automorphism group is $S_{n^m}$.
\end{theorem}

Theorem \ref{thm:n>=5AUT} characterizes the automorphism group of Abelian codes
${\cal C}(n,m,{\cal W})$ for $n\geq 5$. We prove Theorem \ref{thm:n>=5AUT} in
Section \ref{sec:proofofTheoremAUTn>=5}, using group theoretic arguments. In
particular, we first show that $S_n \wr S_m$ is a subgroup of the automorphism
group of $\mathcal{C}(n, m, \mathcal{W})$. Now, $S_n \wr S_m$ is known from
literature to be a maximal subgroup of $S_{n^m}$ for odd $n \ge 5$. We then
identify the codes $\mathcal{C}(n, m, \mathcal{W})$ for which the automorphism group
is $S_{n^m}$, and the remaining cases have $S_n \wr S_m$ as their automorphism group.

\begin{theorem}
    \label{theorem:affine_type_permutation}
    If $\mathcal{W}$ is not one of the boundary cases (described in
    \Cref{note:trivial-codes-case}), then $\AUT(\mathcal{C}(3, m, \mathcal{W})) \le
    \operatorname{AGL}(m, 3)$, the affine general linear group of ${\mathbb F}_3^m$.
    If ${\cal W}$ is one of the boundary cases, then $\AUT(\mathcal{C}(3, m,
    \mathcal{W}))=S_{3^m}$.
\end{theorem}

\Cref{theorem:affine_type_permutation} states that automorphisms of
$\mathcal{C}(3,m,\mathcal{W})$ are always of the form $\sigma(\vect{i}) = A
\vect{i} + \vect{b}$, where $A$ is an invertible matrix. The proof uses the
characterization of maximal subgroups of symmetric groups \cite{LIEBECK1987365}. We
argue that if $S_3 \wr S_m$ is not a maximal subgroup of $S_{3^m}$, then it must
be a proper subgroup of some other maximal subgroup of $S_{3^m}$.
To complete the proof we show that for $m \le 3$, all non-boundary codes have
automorphism group $S_3 \wr S_m$
$\le \AGL(m, 3)$, and for $m \ge 4$, the only maximal subgroup that contains the
product action $S_3 \wr S_m$ is $\AGL(m, 3)$.

Define $O(m, q)$ to be the group of $m \times m$ invertible matrices $A$ over
$\mathbb{F}_q$ such that $A^TA = I$, and $\widetilde{O}(m, q)$ to be the group of
$m \times m$ invertible matrices $A$ over $\mathbb{F}_q$ such that $A^TA = \lambda
I$ for some non-zero $\lambda$. We define
two semidirect product subgroups, $AO(m, q) \triangleq \mathbb{Z}_q^m \rtimes
O(m,q)$ and $A\widetilde{O}(m,q)\triangleq \mathbb{Z}_q^m \rtimes
\widetilde{O}(m,q)$. Clearly, $O(m,q)\leq \widetilde{O}(m,q)$ and hence $AO(m,
q)\leq A\widetilde{O}(m,q)$.

\begin{theorem}
\label{thm:n=3,singleton-weight-set}
    For $m \ge 4$ and $1 \le w \le m$,
    \begin{enumerate}
        \item If $3 \nmid w$ or $m$ is odd, $\AUT({C(3, m, \{w\})}) \le
        \operatorname{AO}(m, 3)$ with equality only in the cases $m = 4, \, w = 2$
        and $m = 5, \, w = 3$.
        \item If $3 \mid w$ and $m$ is even, $\AUT({C(3, m, \{w\})}) \le
        \operatorname{A\widetilde{O}}(m, 3)$ with equality only in the case $m = 4,
        \, w = 3$.
    \end{enumerate}
\end{theorem}

The primary observation required for the proof of Theorem
\ref{thm:n=3,singleton-weight-set} is that if $x$ is a non-zero element of $\mathbb{F}_3$, then $x^2 = 1$. Thus $\Ham(\vect{x}) \equiv
\vect{x}^T\vect{x}$ modulo $3$. This observation gives some necessary conditions
on the automorphisms of $\mathcal{C}(3, m, \{w\})$ which prove Theorem
\ref{thm:n=3,singleton-weight-set}.

\section{The Automorphism Group of Berman Codes:
Proof of Theorem \ref{thm:maintheoremdualberman}}
\label{sec:proofofAutofBerman}

We now prove Theorem \ref{thm:maintheoremdualberman}. $\DB_n(0, m)$ is the
repetition code of length $n^m$, and $\DB_n(m, m)$ is the full space
$\mathbb{F}_2^{n^m}$. The automorphism group of both these codes and their duals is
$S_{n^m}$. The rest of the section is devoted to the case of
$r\in\{1,\hdots,m-1\}$. As $\operatorname{B}_n(r,m)$
and $\DB_n(r,m)$ are dual codes, it is sufficient to show Theorem
\ref{thm:maintheoremdualberman} for $\DB_n(r,m)$. We do this in two stages, given
in Subsections \ref{subsec:hypercubes and dual berman codes} and
\ref{subsec:hamminggraphandAut}. We first obtain some necessary conditions that any
automorphism of $\DB_n(r,m)$ should satisfy, via the minimum weight codewords of
$\DB_n(r,m)$.
Subsequently, we use some graph-theoretic results to complete the proof.

\subsection{Necessary Conditions on Automorphisms of Berman Codes via
Minimum-weight Codewords}
\label{subsec:hypercubes and dual berman codes}

We begin with the notions of an axis-aligned hypercube and its free coordinates.
An axis-aligned (AA) hypercube $H$ of length $n$ in the space $[n]^m$ is defined
as the Cartesian product $H = S_1 \times S_2 \times \dots \times S_m$, where for
each index $i \in [m]$, the set $S_i$ is either the singleton $\{a_i\}$ for some
$a_i \in [n]$, or $S_i = [n]$. The set of free coordinates of the hypercube is
given by $J = \{i\in[m] : S_i = [n]\}$.

Using the characterization of the minimum weight codewords of dual Berman codes \cite[Lemma ~4]{siddheshwar2024recursivesubproductcodesreedmullerlike}, the following observation is easy to see.

\begin{observation}
\label{obs:minwtcodewordsupportHyperplanes}
    The support of each minimum weight codeword of $\DB_n(r,m)$ corresponds to an
    axis-aligned (AA) hypercube in $[n]^m$ with $m-r$ free
    coordinates.
\end{observation}

Let $\mathcal{H}_d$ be the set of AA-hypercubes with $d$ free coordinates, and
$\mathcal{H} = \bigcup_{i = 0}^m \mathcal{H}_i$. We refer to elements of
$\mathcal{H}_1$ as \textit{lines}. We now present some relevant properties of these
hypercubes.

\begin{lemma}
\label{intersection-hypercubes}
    For $H, \, H^{\prime} \in \mathcal{H}$, $H \cap H^{\prime} \in \mathcal{H}$,
    or $H \cap H^{\prime} = \emptyset$.
\end{lemma}

\begin{proof}
    Let $H = S_1 \times S_2 \times \dots \times S_m$, and $H^{\prime} =
    S^{\prime}_1 \times S^{\prime}_2 \times \dots \times S^{\prime}_m$. Thus, $H
    \cap H^{\prime} = \left(S_1 \cap S^{\prime}_1\right) \times \left(S_2 \cap
    S^{\prime}_2\right) \times \dots \times \left(S_m \cap S^{\prime}_m\right)$. If
    $\exists \, i$ such that $S_i \cap S_i^{\prime} = \emptyset$, then $H \cap
    H^{\prime} = \emptyset$. Otherwise $\left|S_i \cap S_i^{\prime}\right| \in \{1,
    n\}, \forall i \in [m]$ which implies $H \cap H^{\prime} \in \mathcal{H}$.
\end{proof}

\begin{lemma}
\label{lemma:how-to-get-small-hypercubes}
    For $1 \le d \le m-2$, and any $H \in \mathcal{H}_d$, there exist distinct
    $H_1,H_2 \in \mathcal{H}_{d+1}$ such that $H_1 \cap H_2 = H$.
\end{lemma}

\begin{proof}
    Let $H = S_1 \times S_2 \times \dots \times S_m$. Since $H \in \mathcal{H}_d$
    and $d \le m-2$, there exist distinct indices $j, k$ such that $|S_j| = |S_k|
    = 1$. Let $H_1 = S_1 \times S_2 \times \dots \times S_{j-1} \times [n] \times
    S_{j+1} \times \dots S_m$, and $H_2 = S_1 \times S_2 \times \dots \times
    S_{k-1} \times [n] \times S_{k+1} \times \dots S_m$. Clearly, $H_1, H_2 \in
    \mathcal{H}_{d+1}$. Further, by construction, we see that $H_1\cap H_2=H$.
\end{proof}

The corollary below easily follows from Lemma \ref{lemma:how-to-get-small-hypercubes}.
\begin{corollary}
\label{lines-as-intersection}
    For any $d \notin \{0, m\}$, for any line $L \in \mathcal{H}_1$, there exist
    $H_1, H_2, \dots, H_{2^{d-1}} \in \mathcal{H}_d$ such that $L = \bigcap_{i =
    1}^{2^{d-1}} H_i$.
\end{corollary}

The following key lemma shows that the automorphism of $\DB_n(r,m)$ preserves
lines in ${\cal H}_1$, for any $1\leq r\leq m-1$.

\begin{lemma}
\label{lemma:lines-preserved}
    Let $f:[n]^m \rightarrow [n]^m \in \AUT \left(\DB_n(r, m)\right)$. For $1 \le
    r \le m-1$, and for each $L \in \mathcal{H}_1$, we have that $f(L) \in
    \mathcal{H}_1$.
\end{lemma}

\begin{proof}
Since $f$ is a code-automorphism, its action on a minimum-weight codeword of
$\DB_n(r,m)$ must result in another minimum-weight codeword in $\DB_n(r,m)$. By
this, and using Observation \ref{obs:minwtcodewordsupportHyperplanes}, we note that
$\forall \, H \in \mathcal{H}_{m-r}$, we have $f(H) \in \mathcal{H}_{m-r}$. By
\Cref{lines-as-intersection} $\exists \, H_1, H_2, \dots, H_{2^{m-r-1}} \in
\mathcal{H}_{m-r}$ such that $L = \bigcap_{i = 1}^{2^{m-r-1}} H_i$. Since $f$ is
a bijective function, we naturally have $f(L) = \bigcap_{i = 1}^{2^{m-r-1}} f(H_i)$
(see Appendix \ref{appendix:bijec_fn_lemma} for a formal proof). Since $f(H_i)$ are
hypercubes, by \Cref{intersection-hypercubes}, $f(L) \in \mathcal{H}$ or $f(L) =
\emptyset$. Moreover, since $f$ is bijective, $\left|f(L)\right| = \left|L\right|$,
so $f(L) \in \mathcal{H}_1$.
\end{proof}

\subsection{Using Automorphisms of the Hamming Graph to complete the proof}
\label{subsec:hamminggraphandAut}

We now use the automorphism group of the Hamming graph to complete the
characterization of the automorphism group of the Berman codes.

The Hamming graph $G_{H}$ has vertices $V_{H} = [n]^m$ and the edge set $E_{H} =
\left\{ \left(\vect{u}, \vect{v}\right) : \forall \vect{u}, \vect{v} \in [n]^m,
d_H\left(\vect{u}, \vect{v}\right) = 1\right\}$. For a graph $G(V,E)$, an
automorphism $\phi:V \rightarrow V$ is a bijection such that, $(u, v) \in E$, if
and only if $(\phi(u), \phi(v)) \in E$. The collection of such automorphisms form a
group under composition. The automorphism group of the Hamming graph has been
completely characterized in literature
\cite{praeger2018permutation-cartesian-decomp}.

\begin{lemma}[{\cite[Corollary 12.4]{praeger2018permutation-cartesian-decomp}}]
    \label{hamming-graph-automorphism}
    $\AUT(G_H) = S_n \wr S_m$.
\end{lemma}

\subsubsection*{Completing the proof of Theorem \ref{thm:maintheoremdualberman}}
Let $f \in \AUT \left(\DB_n(r, m)\right)$, and let $(\vect{u}, \vect{v}) \in
E_{H}$. Observe that $\left(\vect{u}, \vect{v}\right) \in E_{H}$, if and only if
there exists a line $L \in \mathcal{H}_1$ such that $\vect{u}, \vect{v} \in L$.
Let $L$ be such a line. By \Cref{lemma:lines-preserved}, $f\left(\vect{u}\right)$
and $f\left(\vect{v}\right)$ belong to a line, i.e.
$\left(f\left(\vect{u}\right),f\left(\vect{v}\right)\right) \in E_{H}$. Thus, $f$
is an automorphism of $G_{H}$. By \Cref{hamming-graph-automorphism}, this means
$\AUT \left(\DB_n(r, m)\right) \le S_n \wr S_m$, for $1 \le r \le m-1$. Moreover,
it is known that $S_n \wr S_m \le \AUT \left(\DB_n(r, m)\right)$
\cite[Theorem 14]{natarajan2023bermancodesgeneralizationreedmuller} as mentioned in
Subsection \ref{subsubsec:bermancodesandduals}. This shows that $\AUT
\left(\DB_n(r, m)\right) = S_n \wr S_m$, for $1 \le r \le m-1$, which completes the proof of Theorem \ref{thm:maintheoremdualberman}.

\section{Automorphism Group of ${\cal C}(n,m,{\cal W})$, for $n\geq 5$:
Proof of Theorem \ref{thm:n>=5AUT}}
\label{sec:proofofTheoremAUTn>=5}

We now prove Theorem \ref{thm:n>=5AUT}, determining the automorphism groups of several abelian codes constructed via the Discrete
Fourier Transform (DFT) \cite{natarajan2023bermancodesgeneralizationreedmuller,
10.1109/18.165458}, as reviewed in Subsection \ref{subsec:reviewofBermanANDBiDcodes}.
We do this via a sequence of lemmas. We outline these arguments as follows.
\begin{itemize}
\item Lemma \ref{lemma:WRprodalwaysinAUT} shows that the wreath product subgroup
$S_n\wr S_m$ always lies in $\AUT({\cal C}(n,m,{\cal W}))$.
\item Lemma \ref{lemma:WreathProductMaximality}, known from literature
\cite{10.1093/qmath/37.4.419}, shows that any subgroup properly containing $S_n \wr
S_m$ must contain all the even permutations, when $n\geq 5, m\geq 2$.
\item Lemma \ref{lemma:odd-perm-exists} shows that $S_n\wr S_m$ contains odd
permutations, if $n$ is odd.
\item Together with Lemma \ref{lemma:WRprodalwaysinAUT}, Lemmas
\ref{lemma:WreathProductMaximality} and \ref{lemma:odd-perm-exists} imply that the
group $\AUT({\cal C}(n,m,{\cal W}))$ can only be the entire group $S_{n^m}$, or just the wreath product subgroup
$S_n \wr S_m$. We then show the parameter regimes of $\mathcal{W}$ which makes either of
these cases possible, resulting in the proof of Theorem \ref{thm:n>=5AUT}. These
completing arguments are provided in Subsection \ref{subsec:Completingproofofthmn>=5}.
\end{itemize}

We now proceed with the above mentioned lemmas, and provide proofs as needed.

\begin{lemma}
\label{lemma:WRprodalwaysinAUT}
    For any odd $n\geq 3$, and any positive integer $m$, we have $S_n \wr S_m \le
    \AUT \left( \mathcal{C}(n, m, \mathcal{W})\right)$.
\end{lemma}

\begin{proof}
    Consider arbitrary linear codes $\mathcal{C}_1$ and $\mathcal{C}_2$ with equal
    block length. Let $\pi \in \AUT(\mathcal{C}_1) \cap \AUT(\mathcal{C}_2)$. If
    $\vect{c} \in \mathcal{C}_1 \cap \mathcal{C}_2$, then clearly $\pi(\vect{c})
    \in \mathcal{C}_1 \cap \mathcal{C}_2$. Thus, $\AUT(\mathcal{C}_1) \cap
    \AUT(\mathcal{C}_2) \le \AUT(\mathcal{C}_1 \cap \mathcal{C}_2)$. Similarly, if
    $\vect{c}_1+\vect{c}_2 \in \mathcal{C}_1+\mathcal{C}_2$ such that $\vect{c}_1
    \in \mathcal{C}_1, \, \vect{c}_2 \in \mathcal{C}_2$, then
    $\pi(\vect{c}_1+\vect{c}_2) = \pi(\vect{c}_1) + \pi(\vect{c}_2) \in
    \mathcal{C}_1+\mathcal{C}_2$. Thus, $\AUT(\mathcal{C}_1) \cap
    \AUT(\mathcal{C}_2) \le \AUT(\mathcal{C}_1 + \mathcal{C}_2)$. Now, we can use
    the decomposition \cite{natarajan2023bermancodesgeneralizationreedmuller} as
    given in (\ref{eqn:decomposition}),
    $$
    \mathcal{C}(n,m,\mathcal{W}) = \bigoplus_{w \in \mathcal{W}}
    \left(\DB_n(w, m) \cap \B_n(w-1,m)\right).
    $$
    Now, using Theorem \ref{thm:maintheoremdualberman} and our arguments above, we
    thus have that $S_n \wr S_m \le \AUT \left( \mathcal{C}(n, m,
    \mathcal{W})\right)$.
\end{proof}

\begin{corollary}
    $\AUT\left(\mathcal{C}(n, 1,\mathcal{W})\right) = S_n$.
\end{corollary}

\begin{lemma}[{\cite[Theorem B]{10.1093/qmath/37.4.419}}]
    \label{lemma:WreathProductMaximality}
    If $n \ge 5$ and $m \ge 2$,then any subgroup of $S_{n^m}$ properly containing $S_n \wr S_m$, must contain $A_{n^m}$. In other words, for some $G^* \le S_{n^m}$, if it is true that $S_n
    \wr S_m < G^*$, then it holds that $A_{n^m} \leq G^*$.
\end{lemma}

\begin{lemma}
    \label{lemma:odd-perm-exists}
    If $n$ is odd, then $S_n\wr S_m \nleq A_{n^m}$.
\end{lemma}

\begin{proof}
    Consider the natural product action of the wreath product $S_n \wr S_m$ on the
    set $[n]^m$. Let $(\sigma_1,\dots,\sigma_m) \in S_n^m$ be an element of the
    base group. Fix an index $i \in \{1,\dots,m\}$ and suppose that $\sigma_i$ is
    an odd permutation, while $\sigma_j$ is identity for all $j \neq i$. We analyze
    the parity of the induced permutation on $[n]^m$. For each choice of coordinates
    $(x_1,\dots,x_{i-1},x_{i+1},\dots,x_m)$, the permutation $\sigma_i$ acts on
    the $i$th coordinate, producing a copy of $\sigma_i$ on a subset of $[n]^m$ of
    size $n$. There are exactly $n^{m-1}$ such choices, so the induced permutation
    on $[n]^m$ is the product of $n^{m-1}$ disjoint copies of $\sigma_i$. Hence the
    sign of this permutation is $\operatorname{sgn}(\sigma_i)^{\,n^{m-1}}$. Since
    $\sigma_i$ is odd and $n$ is odd, we have $n^{m-1}$ odd, and therefore
    $\operatorname{sgn}(\sigma_i)^{\,n^{m-1}} = -1$. Thus the induced permutation
    on $[n]^m$ is odd. Thus, $S_n \wr S_m \nleq A_{n^m}$.
\end{proof}

\subsection{Completing the proof of Theorem \ref{thm:n>=5AUT}}
\label{subsec:Completingproofofthmn>=5}

Combining the implications of Lemmas \ref{lemma:WRprodalwaysinAUT},
\ref{lemma:WreathProductMaximality} and \ref{lemma:odd-perm-exists}, we now give
the arguments to complete the proof of Theorem \ref{thm:n>=5AUT}. Note that the
only subgroup of $S_{n^m}$ properly containing $A_{n^m}$ is $S_{n^m}$ itself.

Using this fact and \Cref{lemma:WreathProductMaximality} and
\Cref{lemma:odd-perm-exists}, we can say that the only subgroup properly containing
$S_n \wr S_m$ is $S_{n^m}$, when $n$ is odd, which is the case for the code
$\mathcal{C}(n, m, \mathcal{W})$. Thus, the automorphism group of $\mathcal{C}(n,
m, \mathcal{W})$ is either $S_n \wr S_m$ or $S_{n^m}$.

We now show that there are exactly four binary linear codes which have the
automorphism group $S_{n^m}$, and all of these codes are abelian codes of the type
$\mathcal{C}(n, m, \mathcal{W})$ with appropriate ${\cal W}$.

Consider an arbitrary binary linear code $\mathcal{C}$ of block length $N = n^m$,
such that $\AUT(\mathcal{C}) = S_N$. Suppose there is $\vect{c} \in \mathcal{C}$
with $w=\Ham(\vect{c})$. Since $\AUT(\mathcal{C}) = S_N$, this means that every
possible binary vector of weight $w$ must be in $\mathcal{C}$. We must then have
one of two possible cases.
\begin{enumerate}
    \item There are no codewords in $\mathcal{C}$ with non-zero weight that is
    smaller than $N$. The only possible codes are $\{\vect{0}\}$ (in which case
    ${\cal W}=\emptyset$), and the repetition code (in which case ${\cal W}=\{0\}$).
    It is easy to verify in both these cases, $\AUT(\mathcal{C}) = S_N$.
    \item Suppose there is a codeword $\vect{c}\in{\cal C}$ with non-zero weight
    $w<N$. Then, every vector with weight $w$ is present in ${\cal C}$.
    Since $\mathcal{C}$ is linear, it must therefore contain every binary vector of weight
    $2$ (by appropriate sums of pairs of weight-$w$ codewords).

    Thus, $\mathcal{C}$ must either be the single parity check code (for which
    ${\cal W}=\{1,\hdots,m\}$), or $\mathcal{C} = \mathbb{F}_2^N$ (for which
    ${\cal W}=\{0,1,\hdots,m\}$). It is easy to verify in both these cases,
    $\AUT(\mathcal{C}) = S_N$.
\end{enumerate}
Note that we have exhausted all binary codes with $\AUT(\mathcal{C})=S_N$. For
every other choice of ${\cal W}$ apart from the four above, we must thus have $\AUT
\left( \mathcal{C}(n, m, \mathcal{W}) \right) = S_n \wr S_m$. This concludes the
proof of Theorem \ref{thm:n>=5AUT}.

\section{Automorphism Groups of abelian codes ${\cal C}(n,m,{\cal W})$ with $n=3$:
Proofs of Theorems \ref{theorem:affine_type_permutation}
and~\ref{thm:n=3,singleton-weight-set}}
\label{sec:n=3}

We now prove Theorems \ref{theorem:affine_type_permutation} and
\ref{thm:n=3,singleton-weight-set}. Towards that end, we first show the following
lemma, which covers the case of $n=3$ and $m\leq 3$.

\begin{lemma}
\label{lemma:m<=3,n=3}
    For $m \leq 3$, it holds that $\AUT(\mathcal{C}(3, m, \mathcal{W}))
    = S_3 \wr S_m$, except in the boundary cases (described in \Cref{note:trivial-codes-case}) where
    $\AUT(\mathcal{C}(3, m, \mathcal{W})) = S_{3^m}$.
\end{lemma}

\begin{proof}
    If $m=1$, then all possible cases are boundary cases, and hence the statement of the lemma holds true (\Cref{note:trivial-codes-case}). 

    For $m = 2$, we have $S_3 \wr S_2
    \stackrel{(a)}{=} \AUT(\mathcal{C}(3, 2, \{0,1\})) = \AUT(\mathcal{C}(3, 2, \{0\}) \oplus \mathcal{C}(3, 2, \{1\})) \stackrel{(b)}{\ge} \AUT(\mathcal{C}(3, 2, \{0\})) \cap
    \AUT(\mathcal{C}(3, 2, \{1\})) \stackrel{(c)}{=} \AUT(\mathcal{C}(3, 2, \{1\})) \stackrel{(d)}{\ge} S_3 \wr
    S_2$, where $(a)$ holds because of \Cref{thm:maintheoremdualberman}, $(b)$ holds because of the intersection argument (as in proof of \Cref{lemma:WRprodalwaysinAUT}), $(c)$ holds because $\AUT(\mathcal{C}(3, 2, \{0\})) = S_{3^2}$ (\Cref{note:trivial-codes-case}), and $(d)$ holds because of \Cref{lemma:WRprodalwaysinAUT}.
    Thus, $\AUT(\mathcal{C}(3, 2, \{1\})) = S_3 \wr S_2$. Since
    $\AUT(\mathcal{C}(3, 2, \{2\}))$ is a Berman code,
    $\AUT(\mathcal{C}(3, 2, \{2\})) = S_3 \wr S_2$, by Theorem \ref{thm:maintheoremdualberman}. Now, for any ${\cal W}\subseteq\{0,1,2\}$, we see that the set ${\cal W}$ or $\overline{\cal W}$ falls within the cases already discussed above, or is a boundary case in which $\AUT({\cal C}(3,2,{\cal W})=S_{3^2}$ (\Cref{note:trivial-codes-case}). 
    
    Now, for $m = 3$, using similar arguments as in the case of $m=2$, we have $\AUT(\mathcal{C}(3, 3, \{1\})) = \AUT(\mathcal{C}(3, 3, \{3\})) = S_3 \wr S_3$. It can be verified algorithmically that $\AUT(\mathcal{C}(3, 3, \{2\})) = S_3 \wr S_3$ (verification code is available at \cite{github_automorphism_abelian_codes}). Since for $1 \le w \le 3$, $\mathcal{C}(3, 3, \{w\})$ is a subcode of the single parity check code (see Note \ref{note:trivial-codes-case}), it contains only even-weight vectors. Furthermore, because the length of the code is odd, the vectors in $\mathcal{C}(3, 3, \{0, w\}) \backslash \mathcal{C}(3, 3, \{w\})$ have strictly odd weight. As automorphisms preserve Hamming weight, any automorphism of ${\cal C} (3,3,\{0,w\})$ must map ${\cal C} (3,3,\{w\})$ to itself, i.e., $\AUT({\cal C}(3,3,\{0,w\})) \le \AUT({\cal C}(3,3,\{w\})) = S_3 \wr S_3$. Using \Cref{lemma:WRprodalwaysinAUT}, $\AUT({\cal C}(3,3,\{0,w\})) = S_3 \wr S_3$.  Now, for any ${\cal W}\subseteq\{0,1,2,3\}$, we see that the set ${\cal W}$ or $\overline{\cal W}$ falls within the cases already discussed above, or is a boundary case in which $\AUT({\cal C}(3,3,{\cal W})=S_{3^3}$ (\Cref{note:trivial-codes-case}). 
\end{proof}

The proof of \Cref{theorem:affine_type_permutation} relies on the following characterization of maximal subgroups of the symmetric group.
\begin{lemma}[O'Nan-Scott Theorem \cite{LIEBECK1987365}]
\label{lemma:O'Nan-Scott-Theorem}
    If $G$ is one of the maximal subgroups of $S_N$, and $G \neq A_N$, then $G$
    must satisfy one of the following cases:
    \begin{enumerate}
        \item \textit{Intransitive case}: $G = S_M \times S_K$, with $N = M+K$
        and $M \neq K$
        \item \textit{Imprimitive case}: $G = S_M \wr S_K$, with $N = MK$, $M > 1$,
        and $K > 1$
        \item \textit{Affine case}: $G = \AGL(K, p)$, with $N = p^K$ and $p$ is a
        prime
        \item \textit{Group of Diagonal Type}
        \item \textit{Product Action case:} $G = S_M \wr S_K$, with $N = M^K$,
        $M \ge 5$ and $K \ge 1$
        \item \textit{Almost Simple case:} $G$ is an almost simple group acting
        primitively.
    \end{enumerate}
\end{lemma}

Lemmas \ref{lemma:WreathProductMaximality} and \ref{lemma:odd-perm-exists} establish that for odd $n \ge 5$, $S_n \wr S_m$ is a maximal subgroup of $S_{n^m}$. However, this maximality does not extend to the $n=3$ case, as detailed in the following observation.
\begin{observation}
\label{obs:monomial}
A matrix $M$ is said to be a \textit{monomial} matrix if it can be written as $M=
PD$ for some permutation matrix $P$ and an invertible diagonal matrix $D$. Let
$\mathcal{M}_m$ be the set of $m \times m$ monomial matrices over $\mathbb{F}_3$. We observe that any
permutation in $S_3$ can be written as $\sigma (u) = au + b$, where $a, b \in
\mathbb{F}_3$, and $a \neq 0$. Extending this idea, we can see that any permutation
in $S_3 \wr S_m$ can be written as $\sigma (\vect{i}) = A \vect{i} + \vect{b}$,
where $A \in \mathcal{M}_m$.
\end{observation}

We now complete the proof of \Cref{theorem:affine_type_permutation} using Lemmas \ref{lemma:normal-subgroup-of-primitive-group-is-transitive}, \ref{lemma:odd-perm-exists} -- \ref{lemma:O'Nan-Scott-Theorem}, and \Cref{obs:monomial}.

\subsection{Proof of Theorem \ref{theorem:affine_type_permutation}}
\label{subsec:proofoftheoremaffine}

Following arguments in Subsection \ref{subsec:Completingproofofthmn>=5}, we see
that $\AUT(\mathcal{C}(3, m, \mathcal{W}))=S_{3^m}$ if and only if ${\cal W}$ is
one of the boundary cases in \Cref{note:trivial-codes-case}.

We henceforth consider the scenario when ${\cal W}$ is not a boundary case. To
prove Theorem \ref{theorem:affine_type_permutation}, we must show
$\AUT(\mathcal{C}(3, m, \mathcal{W})) \le \operatorname{AGL}(m, 3)$. Now, by the
structure of the wreath product subgroup $S_3 \wr S_m$, it is true that $S_3 \wr
S_m \leq \AGL(m, 3)$. Thus, by using Lemma \ref{lemma:m<=3,n=3} the claim holds
for $m\leq 3$.

We now look at the situation when $m \geq 4$. For convenience, let
$G=\AUT(\mathcal{C}(3, m, \mathcal{W}))$. Note that $G$ is a strict subgroup of
$S_{3^m}$, as ${\cal W}$ is not a boundary case. Further, by
\Cref{lemma:odd-perm-exists}, $G$ is not a subgroup of $A_{3^m}$, which is a
maximal subgroup of $S_{3^m}$. Hence, $G$ must be a subgroup of one of the other
maximal subgroups of $S_{3^m}$, listed in \Cref{lemma:O'Nan-Scott-Theorem}.

We proceed with the elaborate arguments. Note that
$G$ contains the product action of $S_3 \wr S_m$, which is known to be a
primitive transitive group \cite[Lemma 2.7A]{dixon1996permutation}. Hence, $G$ is a
transitive primitive group and must be a subgroup of some primitive maximal subgroup
of $S_{3^m}$, i.e. using \Cref{lemma:O'Nan-Scott-Theorem} $G$ is a subgroup of one
of the following
\begin{enumerate}
    \item $\operatorname{AGL}(m, 3)$
    \item A group of diagonal type
    \item An almost simple group.
\end{enumerate}
Now, we show that the possibilities 2) and 3) cannot occur, thus completing the
proof of Theorem \ref{theorem:affine_type_permutation}.

The degree of groups of diagonal type
is $|T|^{l-1}$ for some non abelian simple group $T$ and some integer $l$
\cite[Theorem 4.1A]{dixon1996permutation}. Since in this case the degree of the
group is $3^m$, $|T|$ must be a power of $3$. If $T$ is simple, that would mean
that the only normal strict subgroup it has is the trivial group $\{e\}$. By the
Feit-Thompson theorem \cite{feit1963solvability}, every group of odd order is
solvable. By definition of a solvable group, this implies that the quotient group
$T/\{e\},$ which is the same as $T$, must be abelian, leading to a contradiction with the definition of group of diagonal type. Thus,
there cannot be any group of diagonal type with degree $3^m$, so possibility 2) cannot happen.

Now consider that $G$ is contained in some almost simple group $G^*$. Thus, there
exists some non-abelian simple group $S$ such that $S \triangleleft G^*$, as per
Definition \ref{defn:almostsimple}. Since $G^*$ acts primitively on $[3]^m$, by
\Cref{lemma:normal-subgroup-of-primitive-group-is-transitive}, $S$ acts transitively
on $[3]^m$. Let $H$ be the stabilizer of any point $\alpha \in [3]^m$ under the action of $S$. By the orbit stabilizer theorem $|S| = |H| \cdot |S(\alpha)|$. Since $S$ acts
transitively, $|S(\alpha)| = 3^m$ i.e. $[S:H] = 3^m$. Subgroups of a simple group
with their index being a prime-power have been characterized in
\cite{GURALNICK1983304}. Using \cite[Theorem ~1]{GURALNICK1983304} it can be seen
that the only possible case where $[S:H] = 3^m$ with $m \ge 4$ is $S = A_{3^m}$.
If $S = A_{3^m}$, then $\AUT(S) = S_{3^m}$
\cite[Theorem ~2.3]{wilson2009finite}. By \Cref{defn:almostsimple} $A_{3^m} <
G^{*} \leq S_{3^m}$. But $A_{3^m}$ is a maximal subgroup of $S_{3^m}$, so $G^* =
S_{3^m}$. Therefore, $G$ is not a subgroup of an almost simple group that is a
proper subgroup of $S_{3^m}$, negating possibility 3).

Thus, the only remaining possibility is 1), i.e., $G \le \operatorname{AGL}(m,3)$.
This completes the proof of Theorem \ref{theorem:affine_type_permutation}.

\subsection{Some useful lemmas towards proving Theorem \ref{thm:n=3,singleton-weight-set}}

By Theorem \ref{theorem:affine_type_permutation}, any automorphism of ${\cal
C}(3,m,{\cal W})$ can be written as a map $\sigma:\vect{i}\to A\vect{i} +
\vect{b}, \forall \vect{i}\in{\mathbb F}_3^m$, for some invertible matrix $A$ in
${\mathbb F}_3^{m\times m}$, and a vector $\vect{b}\in {\mathbb F}_3^m$. This
enables us to analyze the automorphisms of ${\cal C}(3,m,{\cal W})$ (for $m\geq 4$)
in the transform domain, which we illustrate now.

\begin{lemma}
\label{lemma:weightpreservingAx+b}
Consider an invertible matrix $A\in{\mathbb F}_3^{m\times m}$ and vector
$\vect{b}\in{\mathbb F}_3^m$. Let $\sigma$ be a map defined as $\sigma:\vect{i}\to
A\vect{i} + \vect{b}, \forall \vect{i}\in{\mathbb F}_3^m$. Then $\sigma$ defines
an automorphism of ${\cal C}(3,m,{\cal W})$ if and only if $\Ham{(A^T\vect{x})}
\in \mathcal{W}$ for all $\vect{x}$ with $\Ham{(\vect{x})} \in \mathcal{W}$.
\end{lemma}

\begin{proof}
Suppose that $\sigma$ defined by $A$ and $\vect{b}$ is an automorphism of ${\cal
C}(3,m,{\cal W})$. Consider a codeword $\vect{a} \in \mathcal{C}(3, m, \mathcal{W})$,
and let $\vect{a}^{\prime} = (a_{A\vect{i} + \vect{b}}: \vect{i} \in
\mathbb{Z}_3^m)$ be the permuted codeword as per the automorphism $\sigma$. Then
the DFT vector $\hat{{a}}^\prime_{\vect{j}} = \sum_{\vect{i} \in \mathbb{Z}_3^m}
a^\prime_{\vect{i}} \alpha ^{\vect{i}.\vect{j}} = \sum_{\vect{i} \in
\mathbb{Z}_3^m} a_{A^{-1}(\vect{i}-\vect{b})} \alpha ^{\vect{i}.\vect{j}} =
\sum_{\vect{i} \in \mathbb{Z}_3^m} a_{\vect{i}} \alpha ^{(A\vect{i} + \vect{b}).\vect{j}}
= \alpha^{\vect{b}.\vect{j}}\sum_{\vect{i} \in \mathbb{Z}_3^m} a_{\vect{i}}
\alpha ^{\vect{i}.(A^T\vect{j})}
= \alpha^{\vect{b}.\vect{j}} \hat{a}_{A^T\vect{j}}$. Thus,
$\hat{{a}}^\prime_{\vect{j}} = 0$ if and only if $\hat{a}_{A^T\vect{j}} = 0$.
Thus, it must be the case that for all $\vect{x} \in \mathbb{F}_3^m$ such that
$\Ham{(\vect{x})} \in \mathcal{W}$, we must have $\Ham{(A^T\vect{x})} \in
\mathcal{W}$. To prove the converse consider an arbitrary $\sigma: \vect{i}
\rightarrow A \vect{i} + \vect{b}$ where $A$ is invertible and for all $\vect{x}
\in \mathbb{F}_3^m$ with $\Ham{(\vect{x})} \in \mathcal{W}$, $\Ham{(A^T\vect{x})}
\in \mathcal{W}$. Let $\vect{a}$ be a codeword, and $\sigma(\vect{a}) =
\vect{a}^{\prime}$. Let $\pi: \vect{j} \rightarrow A^T \vect{j}$. Since $A^T$ is
full rank, $\pi$ is bijective. Since $\pi(\mathcal{W}) = \mathcal{W}$,
$\pi\left(\overline{\mathcal{W}}\right) = \overline{\mathcal{W}}$
i.e., for any $\vect{j} \in \mathbb{F}_3^m$, $\Ham(\vect{j}) \notin \mathcal{W}$
if and only if $\Ham(A^T\vect{j}) \notin \mathcal{W}$. Thus, $\vect{a}^{\prime}$
is a codeword, i.e., $\sigma$ is an automorphism.
\end{proof}

Thus, following \Cref{lemma:weightpreservingAx+b}, the problem of finding
automorphisms of ${\cal C}(3,m,{\cal W})$ reduces to finding invertible matrices
$A$ such that the map $\vect{x}\to A^T\vect{x}$ ensures
$\Ham(A^T\vect{x})\in{\cal W}$ if $\Ham(\vect{x})\in{\cal W}$.

The following result shows that the affine transformations do not completely lie
within $\AUT({\cal C}(3,m,{\cal W}))$.

\begin{lemma}
\label{lemma:aglNOTFULLYinAUTn=3}
    For any non-boundary set $\mathcal{W}$ (as described in \Cref{note:trivial-codes-case}),
    $\AGL(m, 3) \nleq \AUT(\mathcal{C}\left({3, m, \mathcal{W}}\right))$.
\end{lemma}

\begin{proof}
If the given code ${\cal C}(3,m,{\cal W})$ is a Berman or a dual Berman code, then
by Theorem \ref{thm:maintheoremdualberman} its automorphism group is $S_3\wr S_m$.
If $m=1$ then only boundary cases are possible, and if $2 \le m \le 3$, then in any
non-boundary case, by Lemma \ref{lemma:m<=3,n=3}, the automorphism group is $S_3
\wr S_m$. Now, $S_3\wr S_m$ does not contain all the affine transformations for
$m\geq 2$. For instance, any transformation of the form $\vect{i}\to
A\vect{i}+\vect{b},\forall \vect{i}\in{\mathbb F}_3^m$, where
$A\in{\mathbb F}_3^{m\times m}$ is not a monomial matrix, does not lie in
$S_3\wr S_m$. Thus, the lemma holds in this case.

Now, consider the other scenario, i.e., ${\cal W}$ is a non-boundary case, which
is also not a Berman or dual Berman code. Thus, there exists some $w \in
\mathcal{W} \backslash \{0\}$ such that $w + 1 \notin \mathcal{W}$.

Let $\vect{u}=(u_1,\hdots,u_m) \in \mathbb{Z}_3^m$ be such that $\Ham(\vect{u}) =
w$. Since $w < m$, there exists an index $k \in \{1,\hdots,m\}$ such that $u_k =
0$. Furthermore, since $w \geq 1$, there exists some index $j$ such that $u_j \neq
0$. Consider a matrix in ${\mathbb F}_3^{m\times m}$, given as $A = I_m + E_{kj}$
where $I_m$ is the $m \times m$ identity matrix and $E_{kj}$ is the matrix with a
$1$ at row $k$, column $j$ and zeros elsewhere. Note that $\operatorname{det}(A) =
1$, hence $A$ is invertible. Therefore, $A$ defines an affine transformation
$\sigma$ in ${\mathbb F}_3^m$, given by $\sigma:\vect{i}\to
(A^T)\vect{i},~\vect{i}\in{\mathbb F}_3^m$. Now, consider the vector $\vect{v} =
(v_1,\hdots,v_m)=A\vect{u}$. We easily see the following:
\[
v_i = \begin{cases}
u_i, & \text{if } i \neq k, \\
u_k + u_j, & \text{if } i = k.
\end{cases}
\]
Since $u_k = 0$ and $u_j \neq 0$, it follows that $\Ham{(\vect{v})} = w + 1
\notin \mathcal{W}$. From \Cref{lemma:weightpreservingAx+b}, we therefore see that
the affine transformation defined by $\sigma$ is not in the automorphism group
$\AUT(\mathcal{C}\left({3, m, \mathcal{W}}\right))$. Thus, $\AGL(m, 3) \nleq \AUT(\mathcal{C}\left({3, m,
\mathcal{W}}\right))$ and the proof is complete.
\end{proof}

Theorem \ref{theorem:affine_type_permutation} and Lemma
\ref{lemma:aglNOTFULLYinAUTn=3} guarantee that $\AUT({\cal C}(3,m,{\cal W}))$ is
a proper subgroup of $\AGL(m,3)$, for any non-boundary set ${\cal W}$. Towards
identifying this subgroup, we prove a few more lemmas.

Now, define $H_{2^l} \triangleq \begin{pmatrix} 1 &1 \\ 1 &-1
\end{pmatrix}^{\otimes l}\in{\mathbb F}_3^{2^l}.$ The following lemma captures the
generators of the groups $O(m,3)$. We say a group $G = \left\langle H_1, H_2,
\dots, H_{p}, g_1, g_2, \dots, g_q \right\rangle$, where $H_1, H_2, \dots, H_p$
are subsets of $G$ and $g_1, g_2, \dots, g_q$ are elements in $G$, if
$\left(\bigcup_{i \in [p]}H_i\right) \cup \left(\bigcup_{i \in [q]}
\{g_i\}\right)$ is a generating set of $G$.

\begin{lemma}
    \label{lemma:Om3-generators}
    For $m \le 3$, $O(m, 3) = \left\langle\mathcal{M}_m\right\rangle$, and for
    $m \ge 4$, $O(m, 3) = \left\langle\mathcal{M}_m, \operatorname{diag}(H_4,
    I_{m-4})\right\rangle,$ where $\operatorname{diag}(A,B)$ refers to the block
    diagonal matrix with its block diagonals being occupied by the square matrices
    $A$ and $B$.
\end{lemma}

\begin{proof}
Firstly, we note that ${\cal M}_m\subseteq O(m, q)$. The inclusion ${\cal
M}_m\subseteq O(m, q)$ is because, for any monomial matrix $PD$ ($P$ being a
permutation matrix, and $D$ being an invertible diagonal matrix), we have $D^2=I_m$ (in ${\mathbb F}_3$) and hence
$PDD^TP^T=PD^2P^T=PP^T=I$.

Define $\mathcal{R} \triangleq \left\{R_{\vect{v}} = I_m +
\frac{\vect{v}\vect{v}^T}{\vect{v}^T\vect{v}}: \forall \vect{v}\in{\mathbb F}_3^m
~\text{s.t}~\vect{v}^T\vect{v} \neq 0\right\}$ (${\cal R}$ is said to be the set
of \textit{reflections}). By Cartan-Dieudonne theorem \cite[Chapter
1.7]{QuadraticFormsLam2005}, $O(m, 3) = \langle \mathcal{R} \rangle$. Now, consider
an arbitrary $R_{\vect{v}} \in \mathcal{R}$. Using induction we show that
$R_{\vect{v}}$ can be obtained via compositions of the collection of generators
given in the statement. Without loss of generality, $\vect{v} = (1\,, \dots \, ,
1\,, 0, \, \dots, 0)^T$. If $\vect{v}$ doesn't have this form, then we can always find $\vect{v}^{\prime}
= M \vect{v}$ where $M$ is a monomial matrix and $\vect{v}^{\prime}$ has such a
structure. In this case, $R_{\vect{v}} = M^TR_{\vect{v}^{\prime}}M$. Let
$\Ham(\vect{v}) = k$. We now complete the proof by case-by-case analysis on $k$.

\textit{Case 1, $k = 1$:} $\vect{v} = (1,0,\cdots,0)^T$ so, $R_{\vect{v}} = I +
\frac{\vect{v}\vect{v}^T}{\vect{v}^T\vect{v}} = \operatorname{diag}(-1, 1, 1,
\dots, 1) \in \mathcal{M}_m$.

\textit{Case 2, $k = 2$:} $\vect{v} = (1,1,0,\cdots,0)^T$ so, $R_{\vect{v}} = I + \frac{\vect{v}\vect{v}^T}{\vect{v}^T\vect{v}} =
\operatorname{diag}(P_2, I_{m-2})$, where $P_2=-\begin{pmatrix}0&1\\1&0
\end{pmatrix}$. It is easy to see that $R_{\vect{v}} \in \mathcal{M}_m$, in this
case too.

\textit{Case 3, $k = 3$:} This isn't possible since $\vect{v}^T\vect{v} = 0$.
Note that for $m \le 3$ these are the only cases possible.

Case 4, $k \ge 4$: Assume that for any $\vect{u}$ such that $\vect{u}^T\vect{u}
\neq 0$ and $\Ham(\vect{u}) < k$, $R_{\vect{u}}$ can be generated by the given
generating set. Let $\vect{v} = (1 \,, 1 \,, 1\,, 1 \, | \, \vect{u}^T)^T$. It
is easy to verify that $\vect{w} = \operatorname{diag}(H_4, I_{m-4}) \vect{v} = (1
\,, 0 \,, 0\,, 0 \, | \, \vect{u}^T)^T$. Moreover, $\vect{w}^T\vect{w} =
\vect{v}^T\vect{v}$, and $\Ham(\vect{w}) < k$. By our assumption $R_{\vect{w}}$
can be generated using the given set of generators. $R_{\vect{v}} =
\operatorname{diag}(H_4, I_{m-4}) R_{\vect{w}} \operatorname{diag}(H_4,
I_{m-4})$. Thus, $R_{\vect{v}}$ can be generated using the given generator set.

Hence, for $m \le 3$, $O(m, 3) = \left\langle\mathcal{M}_m\right\rangle$, and for
$m \ge 4$, $O(m, 3) = \left\langle\mathcal{M}_m, \operatorname{diag}(H_4,
I_{m-4})\right\rangle$.
\end{proof}

The following lemma characterizes the generating group of $\widetilde{O}(m,3)$ for $m\geq 4$.

\begin{lemma}
    \label{tilde-Om3-generators}
    Let $m\geq 4$. If $m$ is odd, then $\widetilde{O}(m, 3) = O(m, 3)$. Else, if $m$ is even, $\widetilde{O}(m,
    3) = \left\langle O(m,3), I_{m/2} \otimes H_2\right\rangle$.
\end{lemma}

\begin{proof}
We start with the case of odd $m$. Let $A \in \widetilde{O}(m,3)$. Since we are
operating over $\mathbb{F}_3$, $A^TA = \pm I_m$. Now, if $A^TA = -I$, then
$\operatorname{det}(A^T)\operatorname{det}(A) = \operatorname{det}(-I_m)$, i.e.,
$\left(\operatorname{det}(A)\right)^2 = (-1)^m = -1$ (odd $m$). But $-1$ is not a square in $\mathbb{F}_3$.
Thus, if $A \in \widetilde{O}(m, 3)$, then $A^TA = I_m$. Hence, $O(m, 3) =
\widetilde{O}(m, 3)$.

Now, we will look at the case of even $m$. We have at least one matrix $A =
I_{\frac{m}{2}} \otimes H_2$ such that $A^TA = -I_m$. So for even $m$, $O(m, 3)
< \widetilde{O}(m, 3)$. Consider an arbitrary matrix $B \in \widetilde{O}(m, 3) \backslash O(m, 3)$. Now, $B = A (-AB)$. But,
$(-AB)^T(-AB) = B^TA^TAB = -B^TB = I_m$. Thus, $-AB\in O(m,3)$. Thus, for even $m$, $\widetilde{O}(m, 3) = \left\langle
O(m,3), I_{m/2} \otimes H_2\right\rangle$.
\end{proof}

\begin{observation}
\label{obs:maximality-of-O-in-tilde-O}
For even $m \ge 4$, let $A, B \in \widetilde{O}(m,3) \backslash O(m,3)$, and let
$M = A^{-1}B = -A^TB$ (using $-A^TA = I$).
$M^TM = (-B^TA)(-A^TB) = B^TAA^TB = -B^TB = I$. So, $M \in O(m,3)$. This shows
that there are exactly two cosets of $O(m,3)$ in $\widetilde{O}(m,3)$, i.e.,
$[\widetilde{O}(m,3): O(m,3)] = 2$. By Lagrange's theorem, $O(m,3)$ is a maximal
subgroup of $\widetilde{O}(m,3)$.
\end{observation}

\subsection{Proof of Theorem \ref{thm:n=3,singleton-weight-set}}
\label{subsec:ProofSingletonWeightSet}

For $w = m$, the code is precisely a Berman code and we know the automorphism group
is $S_3\wr S_m$.

Now we will look at the case of $w < m$. Recall that by Theorem
\ref{theorem:affine_type_permutation} and Lemma \ref{lemma:aglNOTFULLYinAUTn=3},
$\AUT({\cal C}(3,m,\{w\}))$ is a proper subgroup of $\AGL(m,3)$. Let $A$ be an $m \times m$ matrix over
$\mathbb{F}_3$. By \Cref{lemma:weightpreservingAx+b}, if $A^T$ corresponds to an
automorphism of $C(3, m, \{w\})$, then for all $\vect{x} \in \mathbb{F}_3^m$ of
weight $w$, we must have $\Ham(A\vect{x}) = w$. We show that this leads to the
condition that $A\in \widetilde{O}(m,3)$. We observe that in $\mathbb{F}_3$, $1^2
= 1$, and $2^2 = 1$. Thus, the standard quadratic form $Q(\vect{x}) \triangleq
\vect{x}^T\vect{x}$ reduces to $Q(\vect{x})=\Ham(\vect{x}) \, \operatorname{mod}
\, 3$. Thus, if $A$ corresponds to an automorphism, we must thus have $Q(\vect{x})
= Q(A\vect{x})$ which should be equal to $\vect{x}^TA^TA\vect{x}$. Let $S =
A^TA$, which is clearly symmetric. Thus, the relationship obtained
$Q(\vect{x})=\vect{x}^T\vect{x} = \vect{x}^TS\vect{x}=Q(A\vect{x})$
becomes equivalent to $\sum_{i \in \supp(\vect{x})} x_i^2 = \sum_{i \in
\supp(\vect{x})} S_{i,i}x_i^2 + 2\sum_{i,j \in \supp(\vect{x}): \,
i<j}S_{i,j}x_ix_j$. By rearranging the equation and using the observation that for
non-zero entries in the vector $x_i^2 = 1$, we get
\begin{align}
    \label{eqn:one}
    \sum_{i \in \supp(\vect{x})} (S_{i,i}-1) - \sum_{i,j \in \supp(\vect{x}):
    \, i<j}S_{i,j}x_ix_j = 0.
\end{align}
This must hold for every vector in $\mathbb{F}_3^m$ of weight $w$.

Now let $k \in \supp({\vect{x}})$. Consider the vector $\vect{y}$ obtained by flipping the
sign of the $k^{\text{th}}$ coordinate of $\vect{x}$, i.e., set $y_k = -x_k$ and
$y_i = x_i$ for all $i \neq k$. Since $\vect{x}$ and $\vect{y}$ have the same support, we have the equation
\begin{align}
\nonumber
0&=\sum_{i \in \supp(\vect{x})} (S_{i,i}-1) - \sum_{i,j \in
\supp(\vect{x}): \, i<j}S_{i,j}y_iy_j\\
\nonumber
&= \sum_{i \in \supp(\vect{x})} (S_{i,i}-1) - \sum_{i,j \in
\supp(\vect{x})\backslash\{k\}: \, i<j}S_{i,j}x_ix_j \\
\label{eqn:two}
&~~~~~~~~~~~~~~~+ \sum_{i \in \supp(\vect{x})\backslash\{k\}} S_{i,k}x_ix_k.
\end{align}
Subtracting (\ref{eqn:two}) from (\ref{eqn:one}), we get $2\sum_{i \in
\supp(\vect{x}) \backslash \{k\}} S_{i,k}x_ix_k = 0$.

Since $2$ and $x_k$ are non-zero, we get the result $\sum_{i \in \supp(\vect{x})
\backslash \{k\}} S_{i,k}x_i = 0$. If at least one of the $S_{i,k}$ in the
summation is non-zero, then by the principle of deferred decision, the linear
equation has at most $2^{w-2}$ solutions (as it involves $w-1$ variables
$x_i:i\in\supp(\vect{x})$, which have to be non-zero). However, this equation
should hold for all possible $\vect{x}$ with a given support and thus must have
$2^{w-1}$ solutions. This contradiction implies that  all the $S_{i,k}$ in the equation are $0$. Since choice of support of
$\vect{x}$ and choice of $k$ are arbitrary, all off-diagonal elements of $S$ are
$0$. Thus, $\sum_{i \in \supp(\vect{x})} (S_{i,i}-1) = 0$. Now, with respect to
the $\vect{x}$ chosen before, consider some $\vect{z}$ such that $\supp(\vect{z})
= \left(\supp(\vect{x}) \backslash \{k\}\right) \bigcup \{k^\prime\}$, where
$k^\prime \notin \supp(\vect{x})$. We can construct such a vector because $w < m$.

Now, $\sum_{i \in \supp(\vect{x})} (S_{i,i}-1) - \sum_{i \in \supp(\vect{z})}
(S_{i,i}-1) = 0$. So, $S_{k,k} = S_{k^\prime, k^\prime}$. Hence, $S = \lambda I$
where $\lambda = 1$ or $\lambda = -1$. Thus, $A \in \widetilde{O}(m, 3)$. Now, if
$\vect{x}^T\vect{x} \neq 0$ and $A^TA = -I$, then $(A\vect{x})^T(A\vect{x}) =
-\vect{x}^T\vect{x}$.
Thus, if $3 \nmid w$, $A \in O(m,3)$.

For most choices of $m$, we now show that not all possible $A \in \widetilde{O}(m,
3)$ can correspond to weight-preserving automorphisms, thus narrowing the choices
for $A$ to some proper subgroup of $\widetilde{O}(m, 3)$. Consider the case $m \ge
6$. By Lemmas \ref{lemma:Om3-generators} and \ref{tilde-Om3-generators}, we have
that $\operatorname{diag}(H_4, I_{m-4})\in \widetilde{O}(m,3)$. Let $w$ be a
non-negative integer such that $w \leq m-3$. Then, we can construct a vector
$\vect{v} = (1,0,0,0|\vect{u}^T)^T$ of weight $w$. Now the weight of
$\operatorname{diag}(H_4, I_{m-4}) \vect{v}$ is $w + 3$. Similarly, for any $w >
m-3$, there exists a pair of vectors $\vect{v},\vect{v}'$ such that
$\Ham(\vect{v})=w-3,\Ham(\vect{v}')=w$ and $\vect{v}'=\operatorname{diag}(H_4,
I_{m-4})\vect{v}$. Thus, $\operatorname{diag}(H_4, I_{m-4})$ cannot be an
automorphism of ${\cal C}(3,m,\{w\})$.

This leaves us with the cases of $m = 4$, and $m = 5$. We will use an observation
from the argument in the case of $m \ge 6$. If $w+3 \le m$ or $w-3 > 0$ then there
is some sequence of monomial operators and $\operatorname{diag}(H_4, I_{m-4})$,
which will take a vector of weight $w$ to a vector of weight $w+3$ or $w-3$. If $m
= 4$, and $w \in \{1, 4\}$ then $\AUT(\mathcal{C}(3, m, \{w\}))$ does not include those permutations defined by the affine maps corresponding to
$\operatorname{diag}(H_4, I_{m-4})$, i.e. $\AUT(\mathcal{C}(3, m,
\{w\})) < AO(m,3)$.
Similarly, if $m = 5$, and $w \in \{1, 2, 4, 5\}$, then $\AUT(\mathcal{C}(3,
m, \{w\})) < AO(m,3)$.

We will now look at the exceptional cases. Consider the case of $m=4$. Let
$\vect{v}$ be a vector of weight $2$. Consider arbitrary $A \in O(m,3)$. As shown
in the first half of the proof $\Ham(A\vect{v}) \equiv w \, (\operatorname{mod} \,
3)$. Since there is no other possible non-zero weight $w^\prime \neq w$ such that
$w^\prime \equiv w (\operatorname{mod} 3)$, for any matrix $A \in O(m, 3)$
$\Ham\left(A\vect{v}\right) = w$. And since $2 \not\equiv 0 (\operatorname{mod} 3)$,
if $A^TA = -I$, $\Ham(A\vect{v}) \equiv \vect{v}^TA^TA\vect{v} \equiv
-\vect{v}^T\vect{v}
\not\equiv w (\operatorname{mod} \, 3)$
which implies, $\Ham(A\vect{v}) \neq \Ham(\vect{v})$. This shows that
$\AUT(\mathcal{C}(3,4,\{2\}))$ is a strict subgroup of $A\widetilde{O}(m,3)$ and contains $AO(m,3)$. Thus using \Cref{obs:maximality-of-O-in-tilde-O},
$\AUT(\mathcal{C}(3, 4, \{2\})) = AO(4,3)$.
Similarly, since there are no non-zero
numbers other than $3$ in the range $1$ to $4$ which are congruent to $3$ mod $3$,
and since $3 \equiv 0 (\operatorname{mod} 3)$,
$\AUT(\mathcal{C}(3, 4, \{3\})) = A\widetilde{O}(4,3)$.
Using the same argument and the fact
that for odd $m$ $O(m, 3) = \widetilde{O}(m, 3)$, we can say
$\AUT(\mathcal{C}(3, 5, \{3\})) = AO(5,3)$.

\section{Conclusion}
\label{sec:conclusion}

We identified the complete automorphism group of Berman codes and their duals.
Further, we presented the complete automorphism groups for the abelian codes $\mathcal{C}(n,
m, \mathcal{W})$ in the case of $n \ge 5$. For $n = 3$ and $m \ge 4$, we do not have a closed form expression of the automorphism group for all possible
weight sets. However, it appears possible to algorithmically compute generators of
the group by reduction to graph automorphism. Graph automorphism is a well-studied problem with worst case quasi-polynomial complexity \cite{MATHON1979131,
babai_graph_isomorphism}, and fast practical algorithms \cite{MCKAY201494}.
Further, future investigations into automorphism based decoders are clearly of
interest.

\bibliographystyle{IEEEtran}
\bibliography{references}

\newpage
\appendices
\section{Useful lemma regarding bijective functions}
\label{appendix:bijec_fn_lemma}

\begin{lemma}
\label{bijective-function-intersection}
Let $f: X \to Y$ be a bijective function. For any family of subsets $X_0, X_1,
\dots, X_{\ell-1} \subseteq X$,
\[
f\left(\bigcap_{i \in [\ell]} X_i\right) = \bigcap_{i \in [\ell]} f(X_i).
\]
\end{lemma}

\begin{proof}
We prove the equality by showing inclusion in both directions.
\begin{enumerate}
    \item Let $y \in f\left(\bigcap_{i\in[\ell]} X_i\right)$. By definition, there
    exists some $x \in \bigcap_{i\in [\ell]} X_i$ such that $f(x) = y$. Since $x$
    belongs to the intersection, $x \in X_i$ for every $i \in [\ell]$.
    Consequently, $f(x) \in f(X_i)$ for every $i$, which implies $y \in
    \bigcap_{i\in [\ell]} f(X_i)$. Note that this direction holds for any function
    $f$.
    \item Let $y \in \bigcap_{i \in [\ell]} f(X_i)$. This implies that for each $i
    \in [\ell]$, there exists an element $x_i \in X_i$ such that $f(x_i) = y$.
    Since $f$ is bijective, it is specifically injective; thus, $f(x_1) = f(x_2) =
    \dots = f(x_{\ell}) = y$ implies that $x_1 = x_2 = \dots = x_{\ell}$. Let this
    unique element be $x$. Since $x$ is equal to each $x_i$, we have $x \in X_i$
    for all $i$, and thus $x \in \bigcap_{i \in [\ell]} X_i$. It follows that $y =
    f(x) \in f\left(\bigcap_{i \in [\ell]} X_i\right)$.
\end{enumerate}
Therefore, $f\left(\bigcap_{i \in [\ell]} X_i\right) = \bigcap_{i \in [\ell]}
f(X_i)$.
\end{proof}

\end{document}